\documentclass{iopjournal}

\usepackage[T1]{fontenc}
\usepackage[utf8]{inputenc}
\usepackage{fix-cm}
\usepackage{lmodern}

\usepackage{amsmath}
\usepackage{amsthm}
\usepackage{latexsym}
\usepackage[normalem]{ulem}

\hypersetup{
  colorlinks=true,
  linkcolor=black,
  citecolor=black,
  urlcolor=black
}

\newtheorem{proposition}{Proposition}

\newtheorem{corollary}{Corollary}

\theoremstyle{definition}

\theoremstyle{remark}
\newtheorem{remark}{Remark}

\begin{document}

\articletype{Paper} %	 e.g. Paper, Letter, Topical Review...

%\title{Oppenheimer-Snyder Collapse in \texorpdfstring{$f(R)$}{f(R)} Gravity : Stalemate or Resolution?}

\title{Oppenheimer–Snyder Collapse in Metric \texorpdfstring{$f(R)$}{f(R)} Gravity: Constraints from Generalized Vaidya Exteriors.}

\author{Soumya Chakrabarti$^1$\orcid{0000-0001-7255-9303},
Apratim Ganguly$^2$\orcid{0000-0001-7394-0755},
Radouane Gannouji$^{3,*}$\orcid{0000-0003-0749-7593},
Chiranjeeb Singha$^2$\orcid{0000-0003-0441-318X}
}

\affil{$^1$Department of Physics, School of Advanced Sciences, Vellore Institute of Technology, Vellore, Tiruvalam Rd, Katpadi, Tamil Nadu 632014, India}

\affil{$^2$Inter-University Centre for Astronomy and Astrophysics (IUCAA), Post Bag 4, Ganeshkhind, Pune 411007, India}

\affil{$^3$Instituto de Física, Pontificia Universidad Católica de Valparaíso, Avenida Universidad 330, Valparaíso, Chile}

\affil{$^*$Author to whom any correspondence should be addressed.}

\email{soumya.chakrabarti@vit.ac.in,
apratim@iucaa.in,
radouane.gannouji@pucv.cl,
chiranjeeb.singha@iucaa.in}

\keywords{Oppenheimer--Snyder collapse, metric $f(R)$ gravity,
junction conditions, generalized Vaidya spacetime,
gravitational collapse, modified gravity}

\begin{abstract}
We study the Oppenheimer--Snyder (OS) collapse problem in metric $f(R)$ gravity by matching a homogeneous dust Friedmann--Lema\^itre--Robertson--Walker (FLRW) interior to a generalized Vaidya exterior across a timelike hypersurface. In metric $f(R)$ gravity, regular matching requires the continuity not only of the induced metric and extrinsic curvature, but also of the Ricci scalar and its normal derivative. {For the homogeneous dust interiors considered here, these additional
conditions generically exclude Ricci-flat exteriors, such as the
Schwarzschild solution.} We show that, for an unrestricted generalized Vaidya exterior, the matching conditions fix the boundary data but do not uniquely determine the bulk extension, leaving open the possibility of a physical resolution of the collapse problem. However, once the exterior matter content is restricted to the generalized Vaidya form, the field equations impose a strong constraint, forcing $f_{,R}$ to be linear in the areal radius, $f_{,R}=A(v)\,r+B(v)$. 
For locally invertible $f_{,R}$ with $f_{,RR}\neq 0$, this sharply reduces the admissible class of exteriors, so that the matching data uniquely determine the exterior solution on each interval where the boundary map is locally invertible. {For branches with $f_{,R}>0$, $f_{,RR}\ne0$, and $f_{,R}$
regular and finite at finite curvature, a globally extended
exterior with finite asymptotic curvature requires $A(v)=0$.
The matching conditions then force constant curvature.
Imposing the null energy condition on both sides leaves only
the vacuum-energy component and a locally maximally symmetric
geometry. Under these assumptions, nontrivial collapse of
spatially flat FLRW perfect-fluid interiors is therefore
excluded within the restricted exterior class.}
\end{abstract}

\section{Introduction}
\label{sec:intro}
The Oppenheimer--Snyder (OS) model remains the canonical idealized description of gravitational collapse, a homogeneous dust-filled FLRW interior matched across a timelike comoving hypersurface to an exterior Schwarzschild geometry~\cite{Oppenheimer:1939ue}. In general relativity (GR), this construction is governed by the Darmois--Israel junction conditions, namely the continuity of the first and second fundamental forms across the matching surface~\cite{Darmois1927,Israel:1966rt}. More general radiating configurations are commonly described by Vaidya-type exteriors, and their matching to collapsing interiors has long served as a standard framework for dissipative collapse~\cite{Vaidya:1951fdr,Santos:1985inc}. {The general-relativistic literature includes systematic models of
evolving radiating fluid spheres
\cite{Herrera:1980zz} and extensive studies of collapse
with heat flow matched to Vaidya exteriors
\cite{BonnorOliveiraSantos1989}.} In particular, generalized Vaidya spacetimes provide a natural extension in which the mass function depends on both the null coordinate and the areal radius, thereby accommodating a broader class of matter sources than pure null dust~\cite{Husain:1995bf,Wang:1998qx}.

The situation is considerably more restrictive in metric $f(R)$ gravity, one of the simplest higher-curvature extensions of GR, widely studied in cosmology, astrophysics, and inflationary model building \cite{Starobinsky:1980te,Sotiriou:2008rp,DeFelice:2010aj,Clifton:2011jh}. In addition to the usual spin-2 graviton, metric $f(R)$ theories propagate an additional scalar degree of freedom (the scalaron), reflecting their equivalence to scalar--tensor theories~\cite{Teyssandier:1983zz}. The viability of such theories requires standard conditions such as $f_{,R} > 0$ and $f_{,RR} > 0$, together with the absence of tachyonic instabilities. {The condition $f_{,RR}>0$ is required to avoid the
Dolgov-Kawasaki instability in the metric formalism
\cite{Faraoni:2006sy,Sawicki:2007tf}.} Because the field equations are fourth order in the metric, the matching problem is no longer governed solely by the Darmois--Israel conditions. For regular matching with $f_{,RR} \neq 0$, one must also impose the continuity of the Ricci scalar and of its normal derivative across the hypersurface~\cite{Deruelle:2007pt,Senovilla:2013vra,Rosa:2023tph,Senovilla:2026fby}. Physically, these additional conditions reflect the requirement that the scalar degree of freedom and its flux remain continuous across the matching surface.

These extra junction conditions severely constrain collapse scenarios. In particular, they obstruct the naive OS construction in metric $f(R)$ gravity whenever one attempts to match a nontrivial homogeneous interior directly to a Ricci-flat exterior~\cite{Senovilla:2013vra,Goswami:2014lxa,Casado-Turrion:2022xkl}. This obstruction is not merely technical but reflects a fundamental incompatibility between the dynamics of the scalar sector and the assumed exterior geometry. A closely related difficulty has also been identified in the context of static compact objects, where the same junction conditions constrain viable neutron star solutions in specific models such as the Starobinsky model~\cite{Ganguly:2013taa}. {More general static interiors can admit smooth matching to a
Schwarzschild exterior \cite{Campbell:2024ouk}. Neutron-star
configurations and compactness bounds have also been studied
with realistic equations of state
\cite{AparicioResco:2016xcm,Fernandez:2025cfp}, together with the geodesic
structure of the resulting spacetimes \cite{RiveroGonzalez:2026pbj}.} This raises a natural question: can relaxing the exterior solution from the standard Vaidya metric to the generalized Vaidya class restore sufficient functional freedom to evade this obstruction, or does the higher-order structure of $f(R)$ gravity reimpose a hidden rigidity? This question is particularly timely in light of recent analyses emphasizing that the standard OS collapse is generically incompatible with the stricter matching conditions and that more general exteriors must be considered~\cite{Casado-Turrion:2022xkl}.

In this work, we address this question by matching a homogeneous FLRW interior to a generalized Vaidya exterior in metric $f(R)$ gravity. Our analysis leads to the following main results. First, without imposing restrictions on the exterior matter content, the matching conditions determine only the restriction of the mass function on the matching hypersurface and do not fix its extension into the bulk, resulting in a non-uniqueness of exterior solutions. Second, once the exterior matter content is restricted to the generalized Vaidya form, the field equations impose a strong rigidity condition: the derivative $f_{,R}$ must be linear in the areal radius, $f_{,R} = A(v)\, r + B(v)$. Third, for locally invertible $f_{,R}$ with $f_{,RR} \neq 0$, this constraint reduces the admissible class of exteriors so drastically that the matching data determine the exterior solution uniquely on each interval where the boundary map is locally invertible. 

However, this mathematical resolution comes at a physical cost. {For the regular $f_{,R}$ branches considered here, $A(v)\neq0$
either prevents continuation to arbitrarily large radius or
requires unbounded curvature along such a continuation, excluding
a global exterior with finite asymptotic curvature.} On the other hand, the branch $A(v) = 0$ places the interior on a constant-curvature sector, which enforces a constant trace of the energy--momentum tensor and excludes nontrivial dust collapse. {Although this branch admits formal constant-trace interiors,
imposing the null energy condition on both sides eliminates
the radiation-like component, leaving only a locally maximally
symmetric geometry with no nontrivial material collapse.}

The resulting picture is therefore more subtle than a simple no-go theorem. Generalized Vaidya exteriors reopen the collapse problem at a formal level by introducing functional freedom in the exterior geometry. However, once physically motivated restrictions on the matter sector are imposed, this freedom collapses into a highly constrained structure, and the standard OS dust collapse remains incompatible with metric $f(R)$ gravity. This suggests that any physically viable realization of gravitational collapse in these theories may require going beyond the generalized Vaidya framework and considering more general exterior matter configurations.

\section{Metric \texorpdfstring{$f(R)$}{f(R)} gravity}
\label{sec:model}

We consider metric $f(R)$ gravity, defined by the action
\begin{align}
\label{eq:SfR}
S=\frac{1}{2\kappa}\int d^4x\,\sqrt{-g}\,f(R)+S_m\,,
\quad \kappa=8\pi G .
\end{align}
{Throughout this work, we use units in which $c = 1$.}
Varying Eq.~\eqref{eq:SfR} with respect to the metric yields the field equations
\begin{align}
f_{,R}R_{\mu\nu}-\frac12 f(R)g_{\mu\nu}
+ \left(g_{\mu\nu}\Box - \nabla_\mu \nabla_\nu\right) f_{,R}
=\kappa\,T_{\mu\nu},
\label{eq:fr_field}
\end{align}
where
\begin{align}
T_{\mu\nu}\equiv -\frac{2}{\sqrt{-g}}\frac{\delta S_m}{\delta g^{\mu\nu}}\,.
\end{align}

Eq.~\eqref{eq:fr_field} implies that, in addition to the spin-2 graviton, metric $f(R)$ gravity propagates an additional scalar degree of freedom encoded in $f_{,R}$. This can be made explicit by considering the trace of the field equations,
\begin{align}
3\Box f_{,R}+f_{,R}\,R-2f(R)=\kappa T,
\label{eq:trace_DK}
\end{align}
which governs the dynamics of the scalaron. Metric $f(R)$ gravity is dynamically equivalent to a scalar--tensor theory in which $f_{,R}$ plays the role of an effective scalar field.

To analyze the propagating degrees of freedom, we linearize the theory around a background metric $\bar g_{\mu\nu}$,
\begin{align}
g_{\mu\nu}=\bar g_{\mu\nu}+h_{\mu\nu},
\end{align}
and decompose the perturbation into its transverse-traceless (TT) component $h^{\text{TT}}_{\mu\nu}$ and the remaining modes. The quadratic action for the TT spin-2 sector takes the schematic form
\begin{align}
S^{(2)}_{\text{spin-2}}
=\frac{f_{,R}(\bar R)}{8\kappa}\int d^4x\,\sqrt{-\bar g}\;
h^{\mu\nu}_{\text{TT}}\,\mathcal{E}\,h^{\text{TT}}_{\mu\nu}
+\cdots,
\label{eq:spin2_quad}
\end{align}
where $\bar R$ is the background scalar curvature and $\mathcal{E}$ denotes the usual second-order Einstein kinetic operator. Requiring a positive kinetic term for the spin-2 graviton leads to the first viability condition
\begin{align}
f_{,R}>0.
\end{align}

The scalar degree of freedom is governed by Eq.~\eqref{eq:trace_DK}. Perturbing the curvature around background curvature scalar $\bar R$ in a region filled with matter, one obtains at leading order a Klein--Gordon type equation
\begin{align}
\left(\Box - M^2(\bar R)\right)\delta R = \frac{\kappa}{3 f_{,RR}(\bar R)}\,\delta T,
\label{eq:KG_DK}
\end{align}
where the effective mass squared is given by
\begin{align}
M^2(\bar R)=\frac{1}{3}\left(\frac{f_{,R}(\bar R)}{f_{,RR}(\bar R)}-\bar R\right).
\label{eq:M2_DK}
\end{align}
The absence of tachyonic instabilities requires \cite{Starobinsky:1980te,Muller:1987hp}
\begin{align}
M^2(R)=\frac{f_{,R}-Rf_{,RR}}{3f_{,RR}}>0.
\label{eq:no_tachyon}
\end{align}
A further constraint arises from the Dolgov--Kawasaki instability \cite{Dolgov:2003px}, which reflects the possibility of rapid growth of curvature perturbations $\delta R$ in matter if the scalar sector has the wrong sign. Avoiding this instability requires from Eq. \eqref{eq:KG_DK}
\begin{align}
f_{,RR}>0.
\label{eq:DK_condition}
\end{align}

In summary, viable metric $f(R)$ models must satisfy the conditions
\begin{align}
f_{,R}>0\,,\qquad f_{,RR}>0\,,\qquad f_{,R}-R f_{,RR}>0.
\end{align}
These conditions ensure a positive, effective gravitational coupling, stability of the scalar sector, and the absence of tachyonic modes.

\section{Junction conditions}
\label{sec:junc-cond}

Let $\Sigma$ be a timelike hypersurface separating an interior region ($-$) from an exterior region ($+$). In a neighborhood of $\Sigma$, we introduce Gaussian normal coordinates $(y,\xi^i)$, where $y \equiv \xi^0$ measures the proper distance along the spacelike unit normal $n^\mu$ to $\Sigma$, and $\xi^i$ ($i=1,2,3$) are intrinsic coordinates on $\Sigma$. The hypersurface is located at $y=0$, and the metric takes the form
\begin{align}
ds^2 = g_{ab}\,d\xi^a d\xi^b = dy^2 + \gamma_{ij}(y,\xi)\,d\xi^i d\xi^j,
\label{eq:gauss_metric}
\end{align}
where $\gamma_{ij}$ is the induced metric (first fundamental form) on the constant-$y$ hypersurfaces, and in particular $\gamma_{ij}\big|_{y=0}$ is the intrinsic metric on $\Sigma$.

The extrinsic curvature (second fundamental form) of $\Sigma$ is defined by
\begin{align}
K_{ij} \equiv \gamma_i{}^{\mu}\gamma_j{}^{\nu}\nabla_\mu n_\nu,
\label{eq:K_def}
\end{align}
where $\gamma_i{}^{\mu}$ is the projector onto $\Sigma$. For an embedding $x^\alpha = x^\alpha(\xi^i)$, this can be written as
\begin{align}
K_{ij}
=-n_{\alpha}\,\Bigl(\frac{\partial^2 x^\alpha}{\partial\xi^i\partial\xi^j}
+\Gamma^{\alpha}{}_{\beta\gamma}
\frac{\partial x^\beta}{\partial\xi^i}\frac{\partial x^\gamma}{\partial\xi^j}\Bigr)\,.
\label{eq:K_coords}
\end{align}

We restrict attention to regular matchings across $\Sigma$, i.e., configurations without distributional sources localized on the hypersurface. In metric $f(R)$ gravity, and assuming $f_{,RR}\neq 0$ on $\Sigma$, the junction conditions take the form \cite{Deruelle:2007pt,Senovilla:2013vra}
\begin{align}
\label{eq:junc1}
\big[\gamma_{ij}\big]^+_- &= 0,\\
\label{eq:junc2}
\big[K_{ij}\big]^+_- &= 0,\\
\label{eq:junc3}
\big[R\big]^+_- &= 0,\\
\label{eq:junc4}
\big[n^\mu\nabla_\mu R\big]^+_- &= 0,
\end{align}
where $[X]^+_- \equiv X^+|_{\Sigma}-X^-|_{\Sigma}$ denotes the jump of any quantity across $\Sigma$, and $n^\mu \nabla_\mu R = \partial_y R$ in Gaussian normal coordinates.

The first two conditions, Eqs.~\eqref{eq:junc1} and \eqref{eq:junc2}, coincide with the standard Darmois--Israel junction conditions of GR~\cite{Darmois1927,Israel:1966rt}. The new requirements in metric $f(R)$ gravity are the continuity of the Ricci scalar and of its normal derivative, Eqs.~\eqref{eq:junc3} and \eqref{eq:junc4}. These arise because the field equations \eqref{eq:fr_field} contain second derivatives of $f(R)$ and therefore involve, in general, fourth derivatives of the metric.

From a physical perspective, these additional conditions reflect the presence of the scalar degree of freedom: regular matching requires not only the continuity of the geometry but also the continuity of the scalaron and of its normal flux across the hypersurface. As a consequence, the space of admissible matchings in metric $f(R)$ gravity is significantly more constrained than in GR, and configurations that are allowed in GR may fail to satisfy these stricter conditions.

\section{Generalized Vaidya exteriors}
\label{sec:gen-vaidya}

The Vaidya spacetime \cite{Vaidya:1951fdr} and its generalizations \cite{Wang:1998qx,Husain:1995bf} are a special case of a more general spherically symmetric line element
\begin{align}
ds^2
&= -e^{2\psi (v,r)}\left[1-\frac{2M(v,r)}{r}\right ]dv^2
+2\,\epsilon\, e^{\psi (v,r)}dv\,dr  + r^2 d\Omega^2,
\qquad (\epsilon=\pm 1),
\label{gg-vaidya}
\end{align}
which can accommodate a combination of Type-I and Type-II matter fields. Indeed, following the Hawking--Ellis algebraic classification of the stress-energy tensor $T^{a}{}_{b}$ \cite{Hawking:1973uf}, a \emph{Type-I} matter field is one for which $T^{a}{}_{b}$ admits a timelike eigenvector $u^{a}$ (and three spacelike eigenvectors). Equivalently, there exists an orthonormal frame in which $T^a{}_{b}$ is diagonal,
\begin{align}
T^a{}_{b}=\mathrm{diag}\!\left(-\rho,\,p_1,\,p_2,\,p_3\right),
\end{align}
so that the matter possesses a local rest frame in which the energy flux vanishes, as in the case of dust or perfect fluids. By contrast, a \emph{Type-II} matter field is characterized by the existence of a \emph{repeated null} eigenvector $k^{a}$ (a double principal null direction), so that $T^{a}{}_{b}$ is not fully diagonalizable in an orthonormal frame. Physically, this corresponds to intrinsically radiative flow with no rest frame that eliminates the energy flux; the canonical example is null dust,
\begin{align}
T_{ab}=\Phi\,k_a k_b, \qquad k^a k_a=0,
\end{align}
see \cite{Hawking:1973uf,Martin-Moruno:2018eil}.

In the metric \eqref{gg-vaidya}, the function $M(v,r)$ acts as a generalized mass function, representing the gravitational energy enclosed within a radius $r$. For $\epsilon=+1$, the null coordinate $v$ denotes the Eddington advanced time, with $r$ decreasing along ingoing null rays at constant $v$. Conversely, $\epsilon=-1$ corresponds to an outgoing null congruence. The condition $\psi(v,r)=0$ selects a particular subclass leading to the generalized Vaidya geometry \cite{Wang:1998qx,Husain:1995bf}. We choose $\epsilon=+1$ and adopt the corresponding form of the metric,
\begin{align}
ds^2
= -\left(1-\frac{2M(v,r)}{r}\right )dv^2
+2dvdr
+r^2 d\Omega^2 .
\label{line-element}
\end{align}
Choosing two null vectors $l_{\mu}$ and $n_{\mu}$ such that
\begin{align}
l_{\mu} &= \delta^0_{\mu}, \qquad
n_{\mu}=\frac{1}{2}\left[1-\frac{2M(v,r)}{r}\right]\delta^0_{\mu}-\delta^1_{\mu}, \qquad
l_{\mu}l^{\mu} = n_{\mu}n^{\mu}=0, \qquad
l_{\mu}n^{\mu}=-1,
\label{normalization}
\end{align}
and assuming general relativity, one can define the corresponding energy-momentum tensor associated to this spacetime, in the form \cite{Lake:1991bff,Wang:1998qx}
\begin{align}
T_{\mu\nu} &= T^{(n)}_{\mu\nu}+T^{(m)}_{\mu\nu},
\label{EMT1}\\
T^{(n)}_{\mu\nu} &= \mu\, l_{\mu}l_{\nu},
\label{EMT2}\\
T^{(m)}_{\mu\nu} &= (\rho+p)(l_{\mu}n_{\nu}+l_{\nu}n_{\mu})+p\, g_{\mu\nu}.
\label{EMT3}
\end{align}
{where $\mu$ denotes the scalar null-fluid density.}
In this setup, the Einstein field equations yield
\begin{align}
\label{GRV1}
\mu &= \frac{2}{\kappa r^2}\frac{\partial M(v,r)}{\partial v}\,,\\
\label{GRV2}
\rho &=\frac{2}{\kappa r^2}\frac{\partial M(v,r)}{\partial r}\,,\\
\label{GRV3}
p &= -\frac{1}{\kappa r}\frac{\partial^{2} M(v,r)}{\partial r^2}.
\end{align}
An energy-momentum tensor written in this form is commonly employed to represent the source supported by a generalized Vaidya metric. The term $T^{(n)}_{\mu\nu}$ represents null radiation propagating along the hypersurfaces $v=\mathrm{const}$ and therefore of type-II, while $T^{(m)}_{\mu\nu}$ is the Type-I, non-null part of the source and may be interpreted as an anisotropic matter sector \cite{Mkenyeleye:2014dwa,Brassel:2021mje}. The full tensor $T_{\mu\nu}=T^{(n)}_{\mu\nu}+T^{(m)}_{\mu\nu}$ is of Type-II whenever $\mu\neq0$. {We first examine the matching problem without prescribing the
exterior stress--energy tensor. From Proposition 5 onward, we
restrict the exterior source to the form (\ref{EMT1}-\ref{EMT3}), without
assuming the general-relativistic relations (\ref{GRV1}--\ref{GRV3}).}

{
\noindent
For the matter tensor (\ref{EMT1}-\ref{EMT3}), the null energy condition requires
\begin{align}
T_{\alpha\beta}k^\alpha k^\beta\geq0
\end{align}
for every null vector \(k^\alpha\). This condition holds if and only if
\begin{align}
\mu\geq0,\qquad \rho+p\geq0,
\end{align}
where \(p\) is the tangential pressure and the radial pressure is \(p_r=-\rho\). These inequalities concern the explicit matter tensor appearing in the metric \(f(R)\) field equations. Their implications for the constant-curvature matched solutions are examined in Proposition \ref{prop:constant-curvature-nec}.
}
\section{Matching FLRW Interior to a Generalized Vaidya Exterior}
\label{sec:match-flrw-gen-vaidya}

\subsection{Geometric Matching and Non-Uniqueness}
\label{subsec:nonuniqueness}

% \paragraph*{\textbf{First and Second Fundamental Forms}:}

We consider the matching of a spatially flat FLRW interior spacetime to a generalized Vaidya exterior across a timelike hypersurface $\Sigma$. The interior metric is
\begin{align}
 ds_{-}^2 = -dt^2 + a^2(t)\left(dr^2 + r^2 d\Omega^2\right),
 \label{inside metric}
\end{align}
while the exterior metric is taken to be of generalized Vaidya form
\begin{align}
 ds_{+}^2 = -\left(1 - \frac{2M(v,r_v)}{r_v}\right)dv^2 + 2\,dv\,dr_v + r_v^2 d\Omega^2,
 \label{outside metric}
\end{align}
where $M(v,r_v)$ is the generalized mass function.

The matching hypersurface $\Sigma$ is defined in the interior coordinates by $r=r_0$, and parametrically in the exterior coordinates as
\begin{align}
r_v = \mathcal{R}(t), \qquad v = \mathcal{T}(t).
\end{align}
The induced metrics on $\Sigma$ are
\begin{align}
 ds_-^2|_\Sigma &= -dt^2 + a^2(t)r_0^2 d\Omega^2, \\
 ds_+^2|_\Sigma &= -\left(V_\Sigma(t)\dot{\mathcal{T}}^2 - 2\dot{\mathcal{T}}\dot{\mathcal{R}}\right)dt^2 + \mathcal{R}(t)^2 d\Omega^2,
\end{align}
where dot represents derivative with respect to $t$, and 
\begin{align}
V_\Sigma(t) \equiv 1 - \frac{2M_\Sigma(t)}{\mathcal{R}(t)}, \qquad M_\Sigma(t) \equiv M(\mathcal{T}(t),\mathcal{R}(t)).
\end{align}
Matching the first fundamental form on $\Sigma$ yields
\begin{align}
&V_\Sigma \dot{\mathcal{T}}^2-2\dot{\mathcal{T}}\dot{\mathcal{R}} = 1, \label{eq:h00}\\
& \mathcal{R}(t)= r_0 a(t). \label{eq:h11}
\end{align}
Imposing continuity of the extrinsic curvature, we compute the unit normals to $\Sigma$. For the interior, the unit normal is
\begin{align}
n_-^\mu = \left(0,\; \frac{1}{a(t)},\; 0,\; 0\right),
\qquad
n^-_\mu = \left(0,\; a(t),\; 0,\; 0\right).
\end{align}
For the exterior, the tangent vector is
\begin{align}
u^\mu \equiv \frac{dx^\mu}{dt} = (\dot{\mathcal{T}},\dot{\mathcal{R}},0,0),
\end{align}
and defining
\begin{align}
\mathcal{R}'(t) \equiv \frac{d\mathcal{R}}{d\mathcal{T}}=\frac{\dot{\mathcal{R}}}{\dot{\mathcal{T}}},
\end{align}
a convenient outward-pointing unit normal (normalized with respect to the exterior metric) is
\begin{align}
n_+^\mu =
\left(\frac{1}{\sqrt{V_\Sigma-2\mathcal{R}'}},\;
\frac{V_\Sigma-\mathcal{R}'}{\sqrt{V_\Sigma-2\mathcal{R}'}},\;0,\;0\right).
\end{align}
For the interior, one finds
\begin{align}
K^-_{tt} &= 0, \qquad
K^-_{\theta\theta} = a(t)\,r_0, \qquad
K^-_{\varphi\varphi} = a(t)\,r_0\,\sin^2\theta.
\end{align}
and for the exterior,
\begin{align}
&K^+_{tt}
= -\frac{1}{\sqrt{V_\Sigma-2\mathcal{R}'}}
\Bigl[
\ddot{\mathcal{R}} - \mathcal{R}'\ddot{\mathcal{T}} +\dot{\mathcal{T}}^2\Bigl((V_\Sigma-3\mathcal{R}')
\Bigl(\frac{M_\Sigma}{\mathcal{R}^2}-\frac{1}{\mathcal{R}}M_{,r_v}\Bigr)
+\frac{1}{\mathcal{R}}M_{,v}\Bigr)
\Bigr]_\Sigma\,, \\
&K^+_{\theta\theta} = \mathcal{R}(t)\frac{V_\Sigma-\mathcal{R}'}{\sqrt{V_\Sigma-2\mathcal{R}'}}\,, \\
&K^+_{\varphi\varphi} = \mathcal{R}(t)\sin^2\theta\frac{V_\Sigma-\mathcal{R}'}{\sqrt{V_\Sigma-2\mathcal{R}'}}\,.
\end{align}
Matching $K_{\theta\theta}$ yields
\begin{align}
V_\Sigma = 1-\dot{\mathcal{R}}^2,\qquad 
\dot{\mathcal{T}} = \frac{1}{1-\dot{\mathcal{R}}}.
\end{align}
Imposing $K^+_{tt}=K^-_{tt}$ gives
\begin{align}
(1-\dot{\mathcal{R}})(1-2 \dot{\mathcal{R}})\left(\ddot{\mathcal{R}}+\frac{M_{\Sigma}}{\mathcal{R}^2}-\frac{M_{,r_v}}{\mathcal{R}} \right)+\frac{M_{, v}}{\mathcal{R}} =0.
\end{align}
Using the previous relations, this reduces to
\begin{align}
\left.M_{,r_v}\right|_{\Sigma}=\frac{M_{\Sigma}}{r_0 a}+r_0^2 a \ddot{a}.
\end{align}

The complete set of Darmois--Israel matching conditions is therefore
\begin{align}
\mathcal{R}(t)  &= r_0 a(t),\label{eq:defR} \\
\dot{\mathcal{T}}(t) &= \frac{1}{1-r_0 \dot a(t)}, \label{eq:defdotT}\\
M_\Sigma(t) &= \frac{r_0^3}{2}a(t)\dot a(t)^2, \label{eq:M-sigma}\\
\left.M_{,r_v}\right|_{\Sigma} &=r_0^2 \Bigl(\frac{1}{2}\dot a(t)^2+a(t) \ddot{a}(t)\Bigr) \label{eq:derM-sigma}.
\end{align}

\paragraph*{\textbf{Standard Vaidya Limit}:}

In the standard Vaidya case one has $\left.M_{, r_v}\right|_{\Sigma}=0$ which reduces Eq. \eqref{eq:derM-sigma} to
\begin{align}
a \ddot{a}+\frac{1}{2} \dot{a}^2=0,
\end{align}
with solution 
\begin{align}
a(t)=a(t_0)\Bigl(1+A(t-t_0)\Bigr)^{2/3} .
\end{align}
It then follows from Eqs.~(\ref{eq:defR}--\ref{eq:M-sigma}) that the evolution of the boundary $\Sigma$ and the interior scale factor are completely determined up to the integration constant $A$; that is, this constitutes a unique family of solutions. For $A<0$, the solution describes collapse.

\paragraph*{\textbf{Generalized Vaidya}:}

For the generalized Vaidya case, $M=M(v,r_v)$. Therefore, for a given interior dynamics $a(t)$, the matching conditions determine only 
\begin{align}
M_\Sigma(t),\qquad \left.M_{,r_v}\right|_{\Sigma}.
\end{align}
This information is specified along a one-dimensional curve in the $\left(v, r_v\right)$ plane and is therefore insufficient to uniquely reconstruct the two-variable function $M\left(v, r_v\right)$ in a neighborhood of $\Sigma$. Infinitely many extensions coincide on $\Sigma$ but differ away from it.

\begin{proposition}
\label{prop:non-uniqueness}
Let $\Sigma$ be the matching hypersurface parametrized by $t \mapsto (\mathcal{T}(t),\mathcal{R}(t))$ in the $(v,r_v)$ plane. 
Prescribing $M$ and finitely many of its partial derivatives on $\Sigma$ does not uniquely determine the function $M(v,r_v)$ in any neighborhood of $\Sigma$.
\end{proposition}

\begin{proof}
The hypersurface $\Sigma$ defines a smooth embedded one-dimensional embedded curve in the $(v,r_v)$ plane. 
In a neighborhood of $\Sigma$, one can introduce local coordinates $(s,n)$ such that $\Sigma$ is given by $n=0$, where $s$ parametrizes the curve and $n$ is a transverse coordinate.

Let $M(s,n)$ be a smooth function representing the mass function in these adapted coordinates. Suppose that, for some fixed integer $N \geq 0$, we prescribe all partial derivatives of $M$ up to total order $N$ on $\Sigma$, i.e.
\begin{align}
\left.\partial_s^p \partial_n^q M\right|_{n=0}
\qquad \text{for all } p,q \geq 0 \text{ with } p+q \leq N.
\end{align}
We now construct a nontrivial deformation of $M$ that preserves all these data on $\Sigma$ but differs away from it. 
Let $\psi(s,n)$ be an arbitrary smooth function which is not identically zero, and define
\begin{align}
\delta M(s,n) = n^{N+1} \psi(s,n).
\end{align}

By construction, $\delta M$ vanishes to order $N+1$ in the transverse direction $n$. Therefore, for all nonnegative integers $p,q$ such that $p+q \leq N$, one has
\begin{align}
\left.\partial_s^p \partial_n^q \delta M\right|_{n=0} = 0.
\end{align}
Indeed, any derivative involving at most $N$ total derivatives cannot remove the overall factor of $n^{N+1}$, so the result still vanishes at $n=0$.

Now define a new function
\begin{align}
\widetilde{M}(s,n) = M(s,n) + \delta M(s,n).
\end{align}
Then, by linearity,
\begin{align}
\left.\partial_s^p \partial_n^q \widetilde{M}\right|_{n=0}
= \left.\partial_s^p \partial_n^q M\right|_{n=0}
\qquad (p+q \leq N).
\end{align}
Thus $\widetilde{M}$ agrees with $M$ on $\Sigma$ up to all derivatives of total order $\leq N$.
However, since $\psi$ is not identically zero, $\delta M \not\equiv 0$ for $n \neq 0$, and therefore
\begin{align}
\widetilde{M}(s,n) \neq M(s,n)
\end{align}
in any neighborhood of $\Sigma$.

Finally, since the coordinate transformation $(v,r_v)\leftrightarrow (s,n)$ is smooth with nonvanishing Jacobian, equality of derivatives up to order $N$ in $(s,n)$ is equivalent to equality of the derivatives $\partial_v^p \partial_{r_v}^q$ up to the same order in $(v,r_v)$. Therefore,
\begin{align}
\left.\partial_v^p \partial_{r_v}^q \widetilde{M}\right|_{\Sigma}
=
\left.\partial_v^p \partial_{r_v}^q M\right|_{\Sigma}
\qquad (p+q \leq N),
\end{align}
while $\widetilde{M} \neq M$ away from $\Sigma$.
This proves that prescribing $M$ and finitely many of its derivatives on $\Sigma$ does not uniquely determine $M(v,r_v)$ in any neighborhood of $\Sigma$.
\end{proof}

\begin{remark}
Even prescribing the full jet on $\Sigma$ does not guarantee uniqueness within the class $C^{\infty}$: there exist nontrivial smooth functions that vanish to infinite order on $\Sigma$ (e.g. $\delta M=e^{-1 / n^2}$ for $n \neq 0, \delta M=0$ for $n=0)$. Uniqueness from jets holds if one strengthens regularity to real-analyticity, or if one supplies additional bulk field equations that propagate the data off $\Sigma$.
\end{remark}

\paragraph*{\textbf{Apparent freedom}:}

In metric $f(R)$ gravity, when the exterior is taken to be the standard Vaidya spacetime $M=M(v)$, the junction conditions overconstrain the system, indicating that homogeneous collapse generically admits no solution. By contrast, within the generalized Vaidya class $M=M\left(v, r_v\right)$, the matching determines only the restriction of $M$ (and one of its derivatives) on $\Sigma$, leaving genuine functional freedom away from $\Sigma$. This additional freedom allows the two extra $f(R)$ junction conditions to be satisfied, making the collapse problem a priori consistent. However, as we show below, this apparent freedom is in fact strongly constrained by the field equations.

\subsection{Matching in \texorpdfstring{$f(R)$}{f(R)} Gravity}

We now impose the junction conditions of Sec.~\ref{sec:junc-cond} to the FLRW--generalized Vaidya matching. For the FLRW interior,
\begin{align}
\label{eq:rminus}
R_- = 6\left(\frac{\ddot a}{a}+\frac{\dot a^2}{a^2}\right),
\end{align}
while for the exterior,
\begin{align}
\label{eq:rplus}
R_+ = \frac{2}{r_v^2}\left(r_v M_{,r_v r_v}+2M_{,r_v}\right).
\end{align}
Continuity of $R$ yields
\begin{align}
\left.M_{,r_v r_v}\right|_{\Sigma}
=r_0\left(\ddot a+\frac{2\dot a^{2}}{a}\right).
\label{eq:ddM-rv}
\end{align}
The condition on the normal derivative~\eqref{eq:junc4} gives
\begin{align}\label{eq:normder-R}
\left.\Bigl(\partial_v R_{+}+(1-\dot{\mathcal{R}})\partial_{r_v}R_{+}\Bigr)\right|_{\Sigma}=0.
\end{align}
Using Eq.~\eqref{eq:rplus}, this leads to
\begin{align}
\label{eq:cond6}
&\left.\Bigl((r_0 a) M_{,r_v r_v v}+2M_{,r_v v}\Bigr)\right|_{\Sigma}
=-\,r_0\frac{1-r_0\dot a}{a}
\Bigl(a^{2}M_{,r_v r_v r_v}\left.-3a\ddot a\Bigr)\right|_{\Sigma}.
\end{align}
The above conditions constrain only derivatives of the mass function $M(v,r_v)$ along the hypersurface $\Sigma$. In particular, they determine $M$ and a finite number of its derivatives on $\Sigma$, but do not provide information about its extension into the bulk spacetime. Consequently, by Proposition~\ref{prop:non-uniqueness}, the additional $f(R)$ junction conditions do not uniquely determine the exterior mass function.

We now analyze how much information about the exterior mass function $M(v,r_v)$ is fixed by the matching conditions on the hypersurface $\Sigma$. In particular, we determine which derivatives of $M$ are fixed by the interior dynamics encoded in the scale factor $a(t)$.

\begin{proposition}
All partial derivatives of the mass function $M(v, r)$ up to total order 3 are fixed on $\Sigma$ by $a(t)$ and its derivatives (up to $\left.a^{(4)}\right)$
\end{proposition}

\begin{proof}
We begin by recalling that the junction conditions determine the value of the mass function 
% \begin{align}
% \label{eq:M-sigma}
% M_\Sigma(t) = \frac{r_0^3}{2} a(t)\dot{a}(t)^2,
% \end{align}
as well as its radial derivative on $\Sigma$ as given in Eqs.~\eqref{eq:M-sigma} and \ref{eq:derM-sigma}.
% \begin{align}
% \label{eq:derM}
% \left.M_{,r_v}\right|_{\Sigma} = r_0^2 \left(\frac{1}{2}\dot{a}(t)^2 + a(t)\ddot{a}(t)\right).
% \end{align}
To determine the derivative $M_{,v}$ on $\Sigma$, we 
% differentiate $M_\Sigma(t)$ along the hypersurface. Using 
use the chain rule to obtain
\begin{align}
\frac{dM_\Sigma}{dt} = \left.M_{,v}\right|_{\Sigma} \dot{\mathcal{T}} + \left.M_{,r_v}\right|_{\Sigma} \dot{\mathcal{R}}.
\end{align}
Substituting the relations~(\ref{eq:defR}-\ref{eq:derM-sigma}) into the above equation, and assuming $\dot{\mathcal{T}}\neq 0$, 
% \begin{align}
% R(t) = r_0 a(t), \qquad \dot{T} = \frac{1}{1 - r_0 \dot{a}(t)}, \qquad \dot{R} = r_0 \dot{a}(t),
% \end{align}
% we substitute the known expressions for $M_\Sigma$ and $M_{,r_v}$ into the above equation. A direct computation 
shows that the resulting expression is identically satisfied if and only if
\begin{align}
\left.M_{,v}\right|_{\Sigma} = 0.
\end{align}

We now proceed to determine second-order derivatives. Differentiating $\left.M_{,r_v}\right|_{\Sigma}$ with respect to $t$ yields
\begin{align}
\frac{d}{dt}\left(M_{,r_v}\right)_\Sigma
=
\left.M_{,r_v v}\right|_{\Sigma} \dot{\mathcal{T}}
+
\left.M_{,r_v r_v}\right|_{\Sigma} \dot{\mathcal{R}}.
\end{align}
From the $f(R)$ junction condition, we already have Eq.~\eqref{eq:ddM-rv}. 
% \begin{align}
% \left.M_{,r_v r_v}\right|_{\Sigma}
% =
% r_0\left(\ddot{a}+\frac{2\dot{a}^2}{a}\right).
% \end{align}
Substituting this together with $\dot{\mathcal{T}}$ from
Eq.~\eqref{eq:defdotT} and
$\dot{\mathcal{R}}=r_0\dot a$, obtained by differentiating
Eq.~\eqref{eq:defR}, one can solve for
$\left.M_{,r_vv}\right|_{\Sigma}$ in terms of $a(t)$ and its derivatives. Similarly, differentiating $\left.M_{,v}\right|_{\Sigma} = 0$ determines $\left.M_{,vv}\right|_{\Sigma}$. We obtain
\begin{align}
\label{eq:Mrv}
\left.M_{, r_v v}\right|_{\Sigma} &= r_0^2\,a\,\Delta\,P,\\
\label{eq:Mvv}
\left.M_{, v v}\right|_{\Sigma} &=-r_0^3\,Q\,\Delta^2\,P,
\end{align}
where we introduced $\Delta=1-r_0\dot a=\dot{\mathcal{T}}^{-1}$, $Q=a\dot a$, and
\begin{align}
P=\dddot{a}+\frac{\dot{a}\ddot{a}}{a}-2\frac{\dot{a}^3}{a^2}.
\end{align}
Proceeding iteratively, one can determine all mixed derivatives up to third order. In particular, by repeated differentiation and use of the junction conditions, one finds
\begin{align}
\label{eq:Mrrr}
&\left.M_{, r_v r_v r_v}\right|_{\Sigma}  = 3\frac{\ddot a}{a}-\frac{3 r_0}{\Delta}P,\\
\label{eq:Mrrv}
&\left.M_{, r_v r_v v}\right|_{\Sigma}  = r_0 \left(1+2r_0 \dot a\right)P,\\
\label{eq:Mrvv}
&\left.M_{, r_v v v}\right|_{\Sigma}  = r_0^2 \Delta \left(\Delta a \dot{P}-r_0\left(a \ddot{a}+3 \dot{a}^2\right) P\right),\\
\label{eq:Mvvv}
&\left.M_{, v v v}\right|_{\Sigma}  = -r_0^3 \Delta^2\left[\Delta \left(2 Q \dot{P}+P \dot{Q}\right)-3r_0 \dot a P\left(a\ddot a+\dot a^2\right)\right].
\end{align}

\end{proof}

{\begin{proposition}
\label{prop:non-uniqueness1}
Suppose that the matching data on $\Sigma$ admit a smooth generalized
Vaidya exterior extension. If the exterior stress--energy tensor is
left unrestricted, this extension is not uniquely determined by the
junction conditions and the exterior field equations.
\end{proposition}}

\begin{proof}
{From Proposition 2, the matching conditions determine the full
third-order jet of the mass function on $\Sigma$. On any interval
where $T(t)$ is locally invertible, the matching surface can be
written as
\begin{equation}
    r_v=r_\Sigma(v),
    \qquad
    r_\Sigma(v)=r_0a\bigl(t(v)\bigr).
\end{equation}
Let $M(v,r_v)$ be a smooth extension realizing these boundary data.
For any smooth function $h(v,r_v)$ that is nonzero on a portion of
$\Sigma$, define
\begin{equation}
    \widetilde M(v,r_v)
    =
    M(v,r_v)
    +\bigl[r_v-r_\Sigma(v)\bigr]^4h(v,r_v).
\end{equation}
The deformation and all its partial derivatives through total order
three vanish on $\Sigma$. Consequently,
\begin{equation}
    \left.
    \partial_v^p\partial_{r_v}^q\widetilde M
    \right|_\Sigma
    =
    \left.
    \partial_v^p\partial_{r_v}^q M
    \right|_\Sigma,
    \qquad p+q\leq3.
\end{equation}
Thus $M$ and $\widetilde M$ satisfy the same junction conditions,
while differing away from the matching surface.}

{To establish that the exterior field equations do not remove this
freedom, we must distinguish between a prescribed source and an
unrestricted stress--energy tensor. For a prescribed source, a
deformation of $M$ must satisfy the field equations with the same
right-hand side. Here, however, the exterior source is not fixed.}

{For the metric $\widetilde g_{\mu\nu}$ associated with
$\widetilde M$, define
\begin{equation}
\begin{aligned}
    \kappa\widetilde T_{\mu\nu}
    \equiv{}&
    f_{,R}(\widetilde R)\widetilde R_{\mu\nu}
    -\frac12 f(\widetilde R)\widetilde g_{\mu\nu}
+
    \left(
    \widetilde g_{\mu\nu}\widetilde\Box
    -\widetilde\nabla_\mu\widetilde\nabla_\nu
    \right)f_{,R}(\widetilde R).
\end{aligned}
\end{equation}
The pair $(\widetilde g_{\mu\nu},\widetilde T_{\mu\nu})$ therefore
satisfies the exterior metric $f(R)$ field equations. Moreover,
$\widetilde T_{\mu\nu}$ is covariantly conserved by the differential
identity obeyed by their geometric left-hand side. Thus, when no
independent matter model is prescribed, the field equations assign
a source to each geometric extension rather than selecting a
unique mass function.}

{The role of this source freedom can be seen explicitly in the
covariant $r_vr_v$ equation,
\begin{equation}
    \kappa T_{r_vr_v}
    =-\partial_{r_v}^2f_{,R}.
\end{equation}
Although the deformation preserves all derivatives of $M$ through
third order on $\Sigma$, it changes the fourth radial derivative:
\begin{equation}
    \left.
    \partial_{r_v}^4(\widetilde M-M)
    \right|_\Sigma
    =24h|_\Sigma.
\end{equation}
Consequently, the corresponding sources satisfy
\begin{equation}
    \kappa
    \left(
    \widetilde T_{r_vr_v}-T_{r_vr_v}
    \right)\big|_\Sigma
    =
    -\frac{48f_{,RR}(R_\Sigma)}{r_\Sigma}\,h|_\Sigma.
\end{equation}
This change is allowed precisely because the exterior stress--energy
tensor is unrestricted. The same matching data therefore admit
distinct exterior geometries supported by distinct source tensors.}

{The construction establishes non-uniqueness with unrestricted
formal sources. It does not establish non-uniqueness for a
prescribed matter model or guarantee that the constructed sources
are physically admissible.}
\end{proof}

\begin{proposition}
In metric $f(R)$ gravity with $f_{,RR} \neq 0$, a homogeneous FLRW interior with non-vanishing Ricci scalar $R_-(t) \neq 0$ cannot be matched to a Ricci-flat exterior across a regular timelike hypersurface $\Sigma$.
\end{proposition}

\begin{proof}
A Ricci-flat exterior satisfies $R_+ = 0$. The continuity condition $[R]_\Sigma = 0$ then implies $R_-(t)=0$ on $\Sigma$. Since the FLRW interior is homogeneous, $R_-(t)$ has no spatial dependence, and therefore the boundary condition forces $R_-(t)=0$ everywhere. This contradicts the assumption $R_-(t)\neq 0$.
\end{proof}

\begin{remark}
This result is consistent with previous analyses in the literature (see, e.g., \cite{Senovilla:2013vra,Casado-Turrion:2022xkl}).
\end{remark}

\begin{remark}
If both the interior and exterior satisfy $R_+=R_-=0$, then the additional junction conditions of metric $f(R)$ gravity are automatically satisfied. In this case, the matching reduces to the standard Darmois--Israel conditions. 

For the spatially flat FLRW interior, the condition $R_-(t)=0$ implies
\begin{align}
a(t)=a(t_0)\Bigl(1+A(t-t_0)\Bigr)^{1/2}.
\end{align}
The matching then determines the exterior mass function uniquely within the scalar-flat branch. Thus, the obstruction discussed above disappears in this special case.
\end{remark}

\begin{remark}
In the generalized Vaidya class, the condition $R_+=0$ constrains the mass function to the form
\begin{align}
M(v,r)=m_1(v)+\frac{m_2(v)}{r}.
\end{align}
% The Schwarzschild solution is recovered by taking $m_2(v)=0$ and $m_1(v)=\mathrm{const}$. This illustrates that Ricci-flat generalized Vaidya spacetimes form a highly restricted subclass.
{This family is scalar-flat, since $R=0$, but its Ricci tensor
generally remains nonzero. The Ricci-flat subclass requires
$m_2(v)=0$ and $m_1'(v)=0$, yielding the Schwarzschild geometry.}

\end{remark}

% \paragraph*{\textbf{Constraints from Field Equations}:}

% We now 
% % examine the implications of the field equations for the exterior geometry. In particular, we 
% show that the functional form of $f_{,R}$ is strongly constrained by the structure of the generalized Vaidya metric and the assumed matter content.

\begin{proposition}\label{constrain-FE}
For the generalized Vaidya metric \eqref{outside metric} with the energy--momentum tensor given by Eqs.~(\ref{EMT1}-\ref{EMT3}), { the rr component in the $(v, r_v, \theta, \varphi)$ coordinate basis} of the field equations in metric $f(R)$ gravity~\eqref{eq:fr_field} implies
\begin{align}
\label{eq:derf-form}
f_{,R}(R(v,r_v))=A(v)\,r_v+B(v),
\end{align}
where $A(v)$ and $B(v)$ are arbitrary functions of the null coordinate $v$.
\end{proposition}

\begin{proof}
The field equations of metric $f(R)$ gravity are given by Eq.~\eqref{eq:fr_field}. We consider the { $(r,r)$ component of these equations in the coordinate system $(v, r_v, \theta, \varphi)$} adapted to the metric \eqref{outside metric}. For the generalized Vaidya metric, one has
\begin{align}
g_{r_v r_v} = 0, \qquad R_{r_v r_v} = 0.
\end{align}
Moreover, for the matter content described by Eqs.~(\ref{EMT1}-\ref{EMT3}), the { covariant $(r,r)$ component} of the energy--momentum tensor also vanishes,
\begin{align}
T_{r_v r_v} = 0.
\end{align}

Substituting these results into Eq.~\eqref{eq:fr_field}, the { $(r,r)$} component reduces to
\begin{align}
-\nabla_{r_v}\nabla_{r_v} f_{,R} = 0.
\end{align}
Since $f_{,R}$ is a scalar function, this becomes
\begin{align}
\partial_{r_v}^2 f_{,R}(R(v,r_v)) = 0,
\end{align}
where we have used the fact that the relevant Christoffel symbols vanish for this component in the chosen coordinates. 
The above equation is an ordinary differential equation in $r_v$, whose general solution is Eq.~\eqref{eq:derf-form}.

{ We emphasize that the condition $T_{r_v r_v}=0$ is a consequence of the particular form of the generalized Vaidya exterior used in this construction. It is not a generic consequence either of spherical symmetry or of metric $f(R)$ theory. Accordingly, the relation $f_{,R}= A(v) r_v + B(v)$ should also be taken as a consequence of the same. If a more general description with $T_{r_v r_v} \neq 0$ is allowed, the $r_v r_v$ component of the field equations no longer leads to $f_{,R,r_v r_v} = 0$. The conclusions based on the affine relation therefore need not
apply when $T_{r_vr_v}\neq0$.  \medskip

Moreover, our analysis has primarily restricted the exterior to a generalized Vaidya with $\psi=0$. For a more general case with $\psi = \psi (v , r_v)$, the covariant $rr$ component of the field equations takes the form
\begin{equation}
F_{,r_v r_v} - \psi_{,r_v}F_{,r_v} - \frac{2\psi_{,r_v}}{r_v}F = -\kappa T_{r_v r_v}.
\label{eq:Frvrv}
\end{equation}
where $F\equiv f_{,R}$.
Therefore, once $\psi_{,r_v}\neq 0$, even a choice of $T_{r_v r_v}=0$ does not in general lead to $F_{,r_v r_v}=0$. The linear form $F = A(v) r_v+B(v)$ is therefore specific to the $\psi_{,r_v}=0$ subclass together with $T_{r_v r_v}=0$.
}

\end{proof}

\begin{remark}
The result shows that the scalar degree of freedom $f_{,R}$ is restricted to be at most linear in the areal radius $r_v$ along each null slice. This represents a strong constraint on the allowed scalar configurations in the exterior spacetime.
\end{remark}

In particular, this restriction implies that the scalar field does not admit localized radial profiles in the generalized Vaidya background, a feature that will play a crucial role in the discussion of physically admissible solutions.

{

\subsubsection{Einstein-frame interpretation}

It may be helpful to include a brief note on the Einstein frame interpretation of Proposition \ref{constrain-FE}. For a locally invertible $f_{,R}$ with $f_{,R}>0$, metric $f(R)$ gravity can be recast as a scalar-tensor theory by introducing the conformal transformation
\begin{equation}
\widetilde{g}_{\mu\nu}=F g_{\mu\nu},~~ F\equiv f_{,R},~~ \phi= \sqrt{\frac{3}{2\kappa}}\ln F .
\end{equation}
In terms of the Einstein-frame metric $\widetilde{g}_{\mu\nu}$ and the scalar field $\phi$, the gravitational part of the action takes the canonical form
\begin{equation}
S_{\rm E} = \int d^4x \sqrt{-\widetilde{g}} \left[\frac{\widetilde{R}}{2\kappa} -\frac{1}{2}\widetilde{g}^{\mu\nu} \partial_\mu\phi \partial_\nu\phi -V(\phi) \right] +S_{\rm m}\left[F^{-1}\widetilde{g}_{\mu\nu},\chi\right],
\end{equation}
where the scalar potential is determined by $f(R)$. For the generalized Vaidya exterior considered here, Proposition \ref{constrain-FE} gives $F(v,r_v)=A(v)r_v + B(v)$. The corresponding Einstein-frame scalar field takes the form
\begin{equation}
\phi(v,r_v) = \sqrt{\frac{3}{2\kappa}} \ln \left[A(v)r_v+B(v)\right].
\end{equation}
Thus, Proposition \ref{constrain-FE} is a restriction on the allowed radial configuration of the scalar. In particular, its radial derivatives are fixed to be
\begin{equation}
\phi_{,r_v} = \sqrt{\frac{3}{2\kappa}} \frac{A(v)}{A(v)r_v+B(v)},\qquad \phi_{,r_v r_v} = -\sqrt{\frac{3}{2\kappa}} \frac{A(v)^2}{\left[A(v)r_v + B(v)\right]^2}.
\end{equation}
The two particular branches identified in the Jordan-frame analysis have a simple Einstein-frame interpretation. 
% \sout{For $A(v)\neq0$, the scalar field necessarily develops the logarithmic radial profile $\phi(v,r_v) \sim \sqrt{\frac{3}{2\kappa}}\ln r_v$ in the limit $r_v\rightarrow\infty$. On the other hand, for $A(v)=0$, one has $F=B(v)$ and $\phi = \sqrt{\frac{3}{2\kappa}}\ln B(v)$ so that the scalar field has no radial dependence on any null slice. Thus, the $A(v)=0$ branch suppresses the radial scalaron profile rather than freezing the scalar degree of freedom completely.}
}
{If $A(v)>0$ and the branch with $F>0$ extends to arbitrarily large
radius, the Einstein-frame scalar develops the asymptotic profile
$\phi\sim\sqrt{3/(2\kappa)}\ln r_v$. Before imposing the junction conditions, the branch $A(v)=0$
eliminates radial dependence but still permits $B(v)$ to vary.
However, if $A(v)=0$ throughout a matched interval, continuity
of the normal derivative gives $B'(v)=0$. Consequently, both
$F$ and the Einstein-frame scalar $\phi$ are constant on that
interval.}

\subsection{Starobinsky Model}

We now consider the Starobinsky model,
\begin{align}
f(R)=R+\alpha R^2,
\end{align}
for which
\begin{align}
f_{,R} = 1 + 2\alpha R.
\end{align}

From Proposition~\ref{constrain-FE}, we know that in the generalized Vaidya geometry the quantity $f_{,R}(R(v,r_v))$ must take the form~\eqref{eq:derf-form}.
% \begin{align}
% f_{,R}(R(v,r_v)) = A(v)\,r_v + B(v),
% \end{align}
% where $A(v)$ and $B(v)$ are arbitrary functions of $v$. 
It follows that the Ricci scalar is given by
\begin{align}\label{Rstarov}
R(v,r_v) = \frac{A(v)\,r_v + B(v) - 1}{2\alpha},
\end{align}
and, using Eq.~\eqref{Rstarov}, we obtain
\begin{align}
M(v,r_v)=\frac{A(v)}{80\alpha}r_v^4+\frac{B(v)-1}{48\alpha}r_v^3+\frac{C(v)}{r_v}+D(v),
\label{eq:massStaro_full}
\end{align}
where $C(v)$ and $D(v)$ are arbitrary functions of $v$ arising as integration functions.

% We now impose the matching conditions on the hypersurface $\Sigma$, where
% \begin{align}
% r_v = R(t) = r_0 a(t), \qquad v = T(t),
% \end{align}
% with
% \begin{align}
% \dot T = \frac{1}{1 - r_0 \dot a}.
% \end{align}

\begin{proposition}
{
For a given interior evolution, the junction conditions
(\ref{eq:defR}-\ref{eq:derM-sigma}) and (\ref{eq:ddM-rv},\ref{eq:normder-R}) uniquely determine the functions
$A(v)$, $B(v)$, $C(v)$ and $D(v)$ on every interval where
$T(t)$ is locally invertible.
}
\end{proposition}

\begin{proof}
We first use the continuity of the Ricci scalar across $\Sigma$,
\begin{align}
R_+(\mathcal{T}(t),\mathcal{R}(t)) = R_-(t),
\end{align}
where $R_-(t)$ is given in Eq.~\eqref{eq:rminus}.
% \begin{align}
% R_-(t) = 6\left(\frac{\ddot a}{a} + \frac{\dot a^2}{a^2}\right).
% \end{align}
Substituting the expression for $R_+$ from Eq.~\eqref{Rstarov} gives
\begin{align}
A(\mathcal{T})\,r_0 a + B(\mathcal{T}) - 1 = 2\alpha R_-(t),
\label{eq:staro1}
\end{align}
where $\mathcal{T}$ is to be regarded as a function of $t$ via Eq. \eqref{eq:defdotT}.

Next, we impose continuity of the normal derivative of the Ricci scalar. Using the matching condition~\eqref{eq:normder-R} and substituting the expression for $R_+$, we obtain
\begin{align}
A(\mathcal{T})+\frac{r_0 a A'(\mathcal{T})+B'(\mathcal{T})}{1-r_0\dot a}=0.
\label{eq:staro2}
\end{align}

Differentiating Eq. \eqref{eq:staro1} with respect to $t$ and using Eq. \eqref{eq:staro2}, we get 
\begin{align}
\label{eq:AT}
A(\mathcal{T})=-\frac{2 \alpha}{1-r_0\dot a} \dot{R}_{-}(t)\,,
\end{align}
and from Eq. \eqref{eq:staro1}
\begin{align}
\label{eq:BT}
B(\mathcal{T})=1+2 \alpha R_{-}(t)-A(\mathcal{T}) r_0 a\,.
\end{align}
Finally, using Eq. \eqref{eq:massStaro_full} and the continuity of the mass function (\ref{eq:M-sigma}, \ref{eq:derM-sigma}) gives
\begin{align}
&\frac{A(\mathcal{T})}{80\alpha}\,r_0^4a^4
+\frac{B(\mathcal{T})-1}{48\alpha}\,r_0^3a^3
+\frac{C(\mathcal{T})}{r_0a}+D(\mathcal{T})=\frac{r_0^3}{2}a\dot a^2\,,\\
&\frac{A(\mathcal{T})}{20\alpha}\,r_0^3a^3
+\frac{B(\mathcal{T})-1}{16\alpha}\,r_0^2a^2
-\frac{C(\mathcal{T})}{r_0^2a^2}
=r_0^2 \Bigl(\frac{\dot a^2}{2}+a \ddot{a}\Bigr),
\end{align}
which determine $C$ and $D$. Therefore the two curvature conditions determine
$A$ and $B$ on $\Sigma$, while the two mass conditions determine $C$ and $D$ on $\Sigma$. Hence all four functions, and thus the mass as well, are fixed on the boundary. 

Since $\dot {\mathcal{T}}=1/(1-r_0\dot a)$, one has that $\dot {\mathcal{T}}$ is well-defined whenever $1-r_0\dot a \neq 0$. Hence $\mathcal{T}$ is locally invertible, and there exists a local inverse $t=t(v)$. Therefore the relations obtained on $\Sigma$ determine $A(\mathcal{T}(t))$, $B(\mathcal{T}(t))$, $C(\mathcal{T}(t))$ and $D(\mathcal{T}(t))$ uniquely as functions of $t$, and thus determine $A(v)$, $B(v)$, $C(v)$ and $D(v)$ uniquely as functions of $v$ on the corresponding interval.

\end{proof}

\begin{remark}

{
The uniqueness obtained in the Starobinsky model does not contradict
Proposition \ref{prop:non-uniqueness1}. That proposition leaves the exterior stress-energy
tensor unrestricted, allowing distinct geometric extensions supported
by distinct formal sources. Here, by contrast, the assumed matter
tensor satisfies $T_{r_vr_v}=0$. The corresponding field equation
forces $f_{,R}=A(v)r_v+B(v)$ and restricts the mass function to the
form \eqref{eq:massStaro_full}. Within this restricted class, the junction conditions
uniquely determine the exterior geometry for a given interior
evolution.
}

% \sout{The complete determination of the exterior solution in the
% Starobinsky model does not contradict the Proposition
% \ref{prop:non-uniqueness1}. In Proposition \ref{prop:non-uniqueness1}, the mass function $M(v,r_v)$ was treated as a generic unknown of a fourth-order system. In that case, the matching conditions fix the boundary $3$-jet, while the restricted field equations do not involve the derivatives $M_{,r_vvvv}$, $M_{,vvvv}$ so that the exterior is not uniquely determined.}

% \sout{By contrast, in the Proposition \ref{constrain-FE} one does not allow an arbitrary matter content. One has $T_{r_v r_v}=0$, and the { covariant $(r,r)$ component} field equation implies Eq. \eqref{Rstarov}. Thus the admissible exterior geometries are restricted to a much smaller class, and within this class the matching conditions are sufficient to determine the solution uniquely.}

\end{remark}

\begin{remark}

Although the solution is uniquely determined, it is generically not physically viable.
{For $A(v)\neq0$, the leading asymptotic terms obtained from the
Starobinsky field equations are
\begin{equation}
    \rho \sim \frac{A(v)^2}{16\kappa\alpha}\,r_v^2,
    \qquad
    p \sim -\frac{A(v)^2}{8\kappa\alpha}\,r_v^2.
\end{equation}
For generic $A(v)A'(v)\neq0$, the null-fluid density behaves as
\begin{equation}
    \mu \sim \frac{A(v)A'(v)}{40\kappa\alpha}\,r_v^3,
\end{equation}
where the prime denotes differentiation with respect to $v$.
If $A(v)A'(v)=0$, subleading terms determine the behavior of $\mu$.}

{
Moreover, for $\alpha>0$,
\begin{equation}
    \rho+p
    \sim-\frac{A(v)^2}{16\kappa\alpha}\,r_v^2<0.
\end{equation}
Thus the formally extended $A(v)\neq0$ branch also violates the
null energy condition at sufficiently large radius.
}

Finally, if $A(v)\neq 0$, the Ricci scalar grows linearly with $r_v$ at large radius,
\begin{align}
R(v,r_v) \sim r_v,
\end{align}
This indicates that the exterior is not sourced by localized matter but instead corresponds to a divergent medium extending to infinity. Imposing the condition $A(v)=0$ removes this divergence, but then Eq.~\eqref{eq:AT} implies that $R_-(t)$ must be constant, which is incompatible with a generic collapsing FLRW interior. Thus, while mathematically consistent solutions exist, they do not correspond to physically acceptable collapse scenarios.
\end{remark}

\subsection{Generic \texorpdfstring{$f(R)$}{f(R)} Models}

We now extend the previous analysis beyond the specific case of the Starobinsky model to a general class of metric $f(R)$ theories. In particular, we assume that $f_{,R}$ is locally invertible along the branch determined by the matching hypersurface $\Sigma$. Denoting the inverse function by
\begin{align}
R = \Phi(f_{,R}), \qquad \Phi = (f_{,R})^{-1},
\end{align}
we analyze the consequences of the matching conditions and field equations for the exterior solution.

\begin{proposition}
For generic $f(R)$ models such that $f_{,R}$ is locally invertible along $\Sigma$ and $f_{,RR}(R_\Sigma)\neq 0$, the matching conditions determine the exterior solution uniquely on every interval where $\mathcal{T}(t)$ is locally invertible.
\end{proposition}

\begin{proof}
From the { covariant $(r,r)$ component} of the field equations~\eqref{eq:fr_field}, we get Eq.~\eqref{eq:derf-form}.
% \begin{align}
% f_{,R}(R(v,r_v)) = A(v)\,r_v + B(v),
% \end{align}
% where $A(v)$ and $B(v)$ are arbitrary functions of $v$.
Since $f_{,R}$ is assumed to be locally invertible, we can express the Ricci scalar as
\begin{align}
R(v,r_v) = \Phi\bigl(A(v)\,r_v + B(v)\bigr).
\label{eq:Rgeneric}
\end{align}

On the other hand, for the generalized Vaidya metric, the Ricci scalar satisfies Eq.~\eqref{eq:rplus}.
% \begin{align}
% R(v,r_v) = \frac{2}{r_v^2}\left(r_v M_{,r_v r_v} + 2M_{,r_v}\right).
% \end{align}
Substituting Eq.~\eqref{eq:Rgeneric} into this relation, we obtain a second-order differential equation in $r_v$ for the mass function $M(v,r_v)$:
\begin{align}
\frac{2}{r_v^2}\left(r_v M_{,r_v r_v} + 2M_{,r_v}\right)
=
\Phi\bigl(A(v)\,r_v + B(v)\bigr).
\end{align}
For each fixed value of $v$, this is an ordinary differential equation in $r_v$. Integrating twice with respect to $r_v$, we obtain
\begin{align}
M(v,r_v) = M_f(v,r_v;A,B) + \frac{C(v)}{r_v} + D(v),
\end{align}
where $M_f$ is a particular solution determined by $A(v)$ and $B(v)$. $C(v)$ and $D(v)$ are arbitrary functions of $v$ arising as integration functions.

We now determine the functions $A(v)$ and $B(v)$ using the matching conditions on $\Sigma$. 
% On the hypersurface $\Sigma$, we have
% \begin{align}
% r_v = R(t) = r_0 a(t), \qquad v = T(t).
% \end{align}
The continuity of the Ricci scalar
% \begin{align}
% R_+(T(t),r_0 a(t)) = R_-(t),
% \end{align}
% where $R_-(t)$ is given in Eq.~\eqref{eq:rminus}.
% \begin{align}
% R_-(t) = 6\left(\frac{\ddot a}{a} + \frac{\dot a^2}{a^2}\right).
% \end{align}
and using Eq.~\eqref{eq:Rgeneric}, we obtain
\begin{align}
\Phi\bigl(A(\mathcal{T})\,r_0 a + B(\mathcal{T})\bigr) = R_-(t).
\end{align}
Applying $f_{,R}$ to both sides yields
\begin{align}
A(\mathcal{T})\,r_0 a + B(\mathcal{T}) = f_{,R}(R_-(t)).
\label{eq:generic1}
\end{align}

Next, we impose the continuity of the normal derivative of the Ricci scalar~\eqref{eq:junc4}.
% \begin{align}
% \left.n^\mu \nabla_\mu R_+\right|_{\Sigma} = 0.
% \end{align}
Using the chain rule,
\begin{align}
\nabla_\mu R_+ = \Phi'(f_{,R})\,\nabla_\mu f_{,R}
= \frac{1}{f_{,RR}(R_+)}\,\nabla_\mu f_{,R},
\end{align}
and since $f_{,RR}(R_\Sigma)\neq 0$, this condition is equivalent to
\begin{align}
\left.n^\mu \nabla_\mu f_{,R}\right|_{\Sigma} = 0.
\end{align}

Using Eq.~\eqref{eq:derf-form} and the expression for the normal vector, one finds that this condition reduces to
\begin{align}
A(\mathcal{T})+\frac{r_0 a A'(\mathcal{T})+B'(\mathcal{T})}{1-r_0\dot a}=0.
\label{eq:generic2}
\end{align}

Eqs.~\eqref{eq:generic1} and \eqref{eq:generic2} form a closed system that determines $A(\mathcal{T}(t))$ and $B(\mathcal{T}(t))$ uniquely in terms of $a(t)$ and its derivatives. In particular, differentiating Eq.~\eqref{eq:generic1} with respect to $t$ and using Eq.~\eqref{eq:generic2}, one obtains
\begin{align}
\label{eq:Ageneric}
A(\mathcal{T}) = -\frac{f_{,RR}(R_-)}{1-r_0\dot a}\,\dot R_-(t),
\end{align}
and therefore
\begin{align}
B(\mathcal{T}) = f_{,R}(R_-(t)) - A(\mathcal{T})\,r_0 a.
\end{align}
Thus, $A(\mathcal{T}(t))$ and $B(\mathcal{T}(t))$ are uniquely determined on $\Sigma$. Since $\dot{\mathcal{T}} \neq 0$ whenever $1-r_0\dot a \neq 0$, the function $\mathcal{T}(t)$ is locally invertible, and therefore $A(v)$ and $B(v)$ are uniquely determined as functions of $v$.

The functions $C(v)$ and $D(v)$ are determined by the matching conditions on the mass function and its radial derivative, Eqs.~\eqref{eq:M-sigma} and \eqref{eq:derM-sigma}, which provide two linear equations for $C(\mathcal{T})$ and $D(\mathcal{T})$. These equations uniquely determine $C(\mathcal{T}(t))$ and $D(\mathcal{T}(t))$.
Finally, since $\mathcal{T}(t)$ is locally invertible, the functions $C(v)$ and $D(v)$ are uniquely determined as functions of $v$. The remaining field equations then determine only the fluid variables
$\mu$, $\rho$ and $p$.
\end{proof}

\begin{remark}
The uniqueness obtained here does not contradict the earlier non-uniqueness results. In the general analysis, the mass function $M(v,r_v)$ was treated as an arbitrary function subject only to the matching conditions and the structure of the field equations, leaving higher-order derivatives unconstrained. In contrast, the restriction \eqref{eq:derf-form} reduces the space of admissible solutions to a finite-dimensional functional family. Within this restricted class, the matching conditions are sufficient to determine the solution uniquely.
\end{remark}

This result shows that, for a broad class of $f(R)$ models, the apparent freedom of the generalized Vaidya exterior is completely eliminated once the energy-momentum tensor is described by Eqs.~(\ref{EMT1}-\ref{EMT3}).

\begin{corollary}
\label{cor:asym}
Given~\eqref{eq:derf-form}  
\begin{align}
f_{,R}(R(v,r_v))=A(v)\,r_v+B(v),
\end{align}
with $f_{,R}>0$, and assume that the solution extends to arbitrarily large $r$. Then either $A(v)=0$ with $B(v)>0$, or $A(v)>0$. In the nontrivial case $A(v)\neq 0$, one has
\begin{align}
f_{,R}(R(v,r))\to +\infty
\qquad (r\to\infty).
\end{align}
Assuming that the inverse branch $\Phi=(f_{,R})^{-1}$ can be continued along the solution for sufficiently large $f_{,R}$ (i.e., on the range of $f_{,R}(R(v,r))$ for large $r$), the asymptotic behavior of the Ricci scalar
\begin{align}
R(v,r)=\Phi(f_{,R}(R(v,r)))
\end{align}
is determined by the limit $f_{,R}\to+\infty$, leading to the following possibilities:
\begin{enumerate}
\item If $f_{,R}$ is bounded from above along the corresponding inverse branch $\Phi$ (equivalently, on the curvature range attained by the solution), then the relation
\begin{align}
F(v,r)=f_{,R}(R(v,r))
\end{align}
cannot hold for arbitrarily large $r$, and therefore no global extension to arbitrarily large $r$ exists.

\item If $f_{,R}(R)\to+\infty$ only when $R\to\infty$, then
\begin{align}
R(v,r)\to\infty
\quad\text{as}\quad r\to\infty.
\end{align}

\item A finite asymptotic limit
\begin{align}
R(v,r)\to R_*,
\qquad
R_*<\infty,
\end{align}
is possible only if
\begin{align}
f_{,R}(R)\to+\infty
\quad\text{as}\quad R\to R_*.
\end{align}
\end{enumerate}
In particular, for generic viable $f(R)$ models in which $f_{,R}$ remains finite at all finite curvature, the branch $A(v)\neq 0$ cannot admit a global extension with finite asymptotic curvature.
\end{corollary}

\begin{proof}
If the branch extends to arbitrarily large $r$ and $f_{,R}>0$, then either $A(v)=0$ (in which case $f_{,R}(v,r)\equiv B(v)>0$), or $A(v)\neq 0$. In the nontrivial case $A(v)\neq 0$, the positivity of $f_{,R}$ for arbitrarily large $r$ forces $A(v)>0$, since otherwise
\begin{align}
f_{,R}(v,r)=A(v)\,r+B(v)\to -\infty
\end{align}
as $r\to\infty$, contradicting $f_{,R}>0$. Hence
\begin{align}
f_{,R}(v,r)\sim A(v)\,r,
\qquad
f_{,R}(v,r)\to +\infty
\quad\text{as}\quad r\to\infty.
\end{align}
If $f_{,R}$ is bounded from above along the corresponding inverse branch $\Phi$ (equivalently, on the curvature range attained by the solution), then $f_{,R}(R(v,r))$ cannot match the unbounded quantity $f_{,R}(v,r)$ for arbitrarily large $r$. Hence no global extension exists.

If $f_{,R}$ becomes unbounded only in the limit $R\to\infty$, then the condition $f_{,R}(v,r)\to+\infty$ implies
\begin{align}
R(v,r)\to\infty
\quad\text{as}\quad r\to\infty.
\end{align}

Finally, if one requires
\begin{align}
R(v,r)\to R_*,
\quad
R_*<\infty,
\end{align}
while
\begin{align}
f_{,R}(v,r)\to+\infty,
\end{align}
then necessarily
\begin{align}
f_{,R}(R)\to+\infty
\quad\text{as}\qquad R\to R_*.
\end{align}
This proves the three alternatives. The final statement follows immediately: if \(f_{,R}\) stays finite at every finite curvature, then the third case is excluded, so for the branch with $A(v)\neq 0$, no global extension with finite asymptotic curvature is possible.
\end{proof}

\begin{remark}
The three alternatives in Corollary~\ref{cor:asym} can be illustrated by simple examples. 

For the first case, one may take
\begin{align}
f(R)=R+\frac{\mu}{R},
\qquad
\mu>0,
\end{align}
for which
\begin{align}
f_{,R}(R)=1-\frac{\mu}{R^2}.
\end{align}
On the physical branch $f_{,R}>0$ with $R>0$, one must have
\begin{align}
R>\sqrt{\mu},
\end{align}
and therefore
\begin{align}
0<f_{,R}(R)<1.
\end{align}
Thus $f_{,R}$ is bounded from above and cannot match a branch satisfying
\begin{align}
f_{,R}(v,r)=A(v)\,r+B(v)\to+\infty
\quad\text{as}\quad r\to\infty.
\end{align}
This realizes the first alternative.

For the second case, one may take the Starobinsky model which has already been discussed.

Finally for the third case, one may take
\begin{align}
f(R)=-R_*\ln\!\left(1-\frac{R}{R_*}\right),
\qquad
R<R_*,
\end{align}
with $R_*>0$. Then
\begin{align}
f_{,R}(R)=\frac{1}{1-R/R_*},
\end{align}
so that $f_{,R}(R)\to+\infty$ as $R\to R_*<\infty$. The inverse relation is
\begin{align}
R=R_*\left(1-\frac{1}{f_{,R}}\right).
\end{align}
Therefore, if
\begin{align}
f_{,R}(v,r)=A(v)\,r+B(v)\to+\infty,
\end{align}
one finds
\begin{align}
R(v,r)\to R_*.
\end{align}
This realizes the third alternative.

These examples are meant only as simple prototypes of the three possible asymptotic behaviors. In particular, the third example is not viable in the usual sense precisely because $f_{,R}$ diverges at finite curvature.
\end{remark}

{
\begin{remark}
The Hu-Sawicki model \cite{Hu:2007nk} provides an example of the first case in Corollary~\ref{cor:asym}. In its standard form, the model is given by
\begin{equation}
f(R) = R-m^2 \frac{c_1 (R/m^2)^n}{c_2 (R/m^2)^n+1},
\end{equation}
and hence
\begin{equation}
f_{,R} = 1- \frac{n c_1 (R/m^2)^{n-1}}
{\left[c_2(R/m^2)^n+1\right]^2}.
\end{equation}
It is straightforward to see that $f_{,R}\to 1$ in the limit $R\rightarrow\infty$. On the other hand the generalized Vaidya construction leads to $f_{,R}(v,r) \sim A(v)r \to +\infty$ in the limit $(r\to\infty)$. Therefore, the $A(v)\neq0$ branch cannot be extended to an arbitrarily large $r$ for the Hu-Sawicki model.
\end{remark}
}

\begin{remark}
Another illustration of the third case in Corollary \ref{cor:asym} is provided by the power-law model \cite{Clifton:2005aj}
\begin{align}
f(R)=R^{1+\delta}.
\end{align}
In this case,
\begin{align}
f_{,R}(R)=(1+\delta)R^\delta.
\end{align}
Hence, on the branch singled out by Proposition \ref{constrain-FE},
\begin{align}
(1+\delta)R^\delta=A(v)\,r+B(v).
\end{align}
In the nontrivial case $A(v)\neq 0$, one obtains asymptotically
\begin{align}
R(v,r)\sim \left(\frac{A(v)}{1+\delta}\right)^{1/\delta} r^{1/\delta}.
\end{align}
If one restricts to
\begin{align}
-1<\delta<0,
\end{align}
then
\begin{align}
R(v,r)\to 0
\qquad (r\to\infty),
\end{align}
while
\begin{align}
f_{,R}(R)=(1+\delta)R^\delta \to +\infty
\qquad \text{as} \qquad R\to 0^+.
\end{align}
Therefore this model realizes precisely the third case of Corollary \ref{cor:asym}, with the finite asymptotic value
\begin{align}
R_*=0.
\end{align}
The corresponding mass does not have a single universal large-$r$ behavior throughout the whole interval $-1<\delta<0$. Using Eq.~\eqref{eq:rplus} one finds that, for generic $\delta\neq -1/4,-1/3$, the mass contains the particular term
\begin{align}
M_{\text{part}}(v,r)\propto r^{3+1/\delta},
\end{align}
but this term is asymptotically dominant only when
\begin{align}
-1<\delta<-\frac13.
\end{align}
For
\begin{align}
\delta=-\frac13,
\end{align}
the behavior becomes logarithmic,
\begin{align}
M(v,r)\sim \ln r,
\end{align}
while for
\begin{align}
-\frac13<\delta<0
\end{align}
the particular term decays and the mass approaches a finite limit, up to subleading corrections. Thus the relevant asymptotic feature is not a universal divergence of $M$, but rather the fact that the curvature tends to the finite value $R_*=0$ while $f_{,R}$ diverges there.\\
However, this branch is physically excluded. Indeed, for
\begin{align}
-1<\delta<0,
\end{align}
one has
\begin{align}
f_{,RR}(R)=\delta(1+\delta)R^{\delta-1}<0,
\end{align}
so the Dolgov--Kawasaki stability condition \eqref{eq:DK_condition} is violated. Moreover, Clifton and Barrow \cite{Clifton:2005aj} found that observational constraints confine the allowed values to a tiny neighborhood of general relativity, namely
\begin{align}
0\le \delta < 7.2\times 10^{-19},
\end{align}
thereby excluding the negative-$\delta$ branch altogether. Hence, although the model provides a concrete example of case~3, it is not physically viable.
\end{remark}

\medskip

The corollary shows that the asymptotic behavior of the exterior solution is entirely controlled by the large-$f_{,R}$ behavior of the inverse function $\Phi=(f_{,R})^{-1}$. In particular, for viable $f(R)$ models in which $f_{,R}$ remains finite at all finite curvature, the branch $A(v)\neq 0$ necessarily leads either to divergent curvature or to the absence of a global exterior solution. 

% \sout{This result has a direct physical interpretation. Since $f_{,R}$ represents the scalar degree of freedom in metric $f(R)$ gravity, the relation~\eqref{eq:derf-form} forces the scalaron to grow without bound in the nontrivial branch $A(v)\neq 0$. This, in turn, leads to an unbounded growth of curvature and matter variables at large radius, indicating that the exterior is not sourced by a localized distribution but rather by a divergent medium extending to infinity. Consequently, this branch cannot describe a physically acceptable exterior spacetime.}

{ We note that the preceding analysis is a statement about the asymptotic extension of the exterior. For $A(v)\neq 0$, the relation
$f_{,R}(v,r)=A(v)r+B(v)$ requires $f_{,R}\rightarrow +\infty$ as $r \rightarrow \infty$ on any branch satisfying $f_{,R}>0$ and extending to arbitrarily large radius. For viable $f(R)$ models in which $f_{,R}$ remains finite at every finite value of the Ricci scalar, this is incompatible with a global exterior having finite asymptotic curvature. Equivalently, the $A(v)\neq0$ branch either fails to admit such a global extension or requires the curvature to become unbounded at large radius. This conclusion does not exclude the existence of local solutions or exterior solutions defined only over a finite radial domain. It does not constitute a general no-go theorem for all possible non-asymptotically-flat or non-localized exterior configurations. It simply shows that, within the generalized Vaidya class, the nontrivial $A(v)\neq 0$ branch cannot provide a globally extended exterior with finite asymptotic curvature for the class of viable $f(R)$ models considered.

\begin{remark}
To obtain a physically acceptable exterior, one is naturally led to consider the branch $A(\mathcal{T})=0$. Assuming $f_{,RR}(R_-(t))\neq 0$, Eq.~\eqref{eq:Ageneric} implies that the interior Ricci scalar is constant, $R_-(t)=R_0$. For dust, $p=0$ and hence $T^-=-\rho$. Therefore, the constant-trace condition requires $\rho = \mathrm{constant}$. Energy-momentum conservation leads to $3H\rho=0$, so that, for non-zero dust density, $H=0$ and hence $\dot a=0$. Thus nontrivial homogeneous dust collapse is excluded. 

However, this does not rule out collapse altogether. Instead, it restricts the interior to a constant-curvature branch. In this case, $f_{,R}(R_0)$ is constant, the derivative terms in the field equations vanish, and the dynamics reduce to
\begin{align}
f_{,R}(R_0)\,G^-_{\mu\nu}+\frac12\bigl(f_{,R}(R_0)R_0-f(R_0)\bigr)g^-_{\mu\nu}
=\kappa\,T^-_{\mu\nu},
\end{align}
i.e., Einstein equations with an effective cosmological constant (for $f_{,R}(R_0)\neq 0$).

The trace Eq.~\eqref{eq:trace_DK} gives
\begin{align}
\kappa\,T^- = f_{,R}(R_0)R_0-2f(R_0),
\end{align}
so the trace of the interior energy-momentum tensor is constant. Under FLRW symmetry, this implies
\begin{align}
T^-=-\rho+3p=T_0,
\end{align}
leading to the equation of state
\begin{align}
p=\frac{\rho+T_0}{3}.
\end{align}
Using energy--momentum conservation,
\begin{align}
\dot{\rho}+3H(\rho+p)=0\,,
\end{align}
it follows that
\begin{align}
\dot{\rho}+H(4\rho+T_0)=0,
\end{align}
and therefore
\begin{align}
\rho(a)=\frac{C}{a^4}-\frac{T_0}{4},
\quad
p(a)=\frac{C}{3a^4}+\frac{T_0}{4},
\end{align}
for some constant $C$. Thus, the admissible matter content consists of a radiation-like component plus an effective vacuum contribution.
{The constant-curvature branch admits formal constant-trace
interiors. However, imposing the null energy condition on
both sides eliminates the radiation-like component, leaving
only a locally maximally symmetric geometry with no
nontrivial material collapse.}

\end{remark}
}

\begin{proposition}[{Matched exterior on the constant-curvature branch}]
\label{prop:constant-curvature-exterior}

{
Consider a spatially flat FLRW perfect-fluid interior satisfying
the metric $f(R)$ field equations with constant Ricci scalar
$R^-=R_0$, where $F_0\equiv f_{,R}(R_0)>0$.
Suppose its comoving boundary $r=r_0$ admits a regular matching
to a generalized Vaidya exterior with $\psi=0$ on the same
constant-curvature branch, $R^+=R_0$.
Writing the interior matter variables as
\begin{align}
    \rho_- &= \frac{C}{a^4}-\frac{T_0}{4},
    &
    p_- &= \frac{C}{3a^4}+\frac{T_0}{4},
    \qquad
    T_0=\frac{F_0R_0-2f(R_0)}{\kappa},
\end{align}
the matching conditions uniquely determine the exterior mass
function on the matched interval:
\begin{align}
    M(v,r_v)
    =\frac{R_0r_v^3}{24}
     +\frac{\kappa C r_0^4}{6F_0r_v}.
    \label{eq:constant-curvature-matched-mass}
\end{align}
In particular, the exterior is independent of $v$.}
\end{proposition}

\begin{proof}
{
Since $R=R_0$, the field equations reduce to
\begin{align}
    F_0G_{\mu\nu}
    +\frac{1}{2}\bigl[F_0R_0-f(R_0)\bigr]g_{\mu\nu}
    =\kappa T_{\mu\nu}.
\end{align}
The interior Friedmann equation therefore gives
\begin{align}
    H^2
    &=\frac{R_0}{12}
      +\frac{\kappa C}{3F_0a^4},
    \qquad H=\frac{\dot a}{a},
    \\
    \dot a^2
    &=\frac{R_0a^2}{12}
      +\frac{\kappa C}{3F_0a^2}.
    \label{eq:constant-curvature-first-integral}
\end{align}
For the exterior, the condition $R^+=R_0$ integrates to
\begin{align}
    M(v,r_v)
    =\frac{R_0r_v^3}{24}
     +\frac{C_1(v)}{r_v}+C_2(v).
\end{align}
On the boundary $r_\Sigma(t)=r_0a(t)$, continuity of the
extrinsic curvature requires
\begin{align}
    M|_\Sigma
    &=\frac{r_0^3}{2}a\dot a^2,
    \\
    M_{,r_v}|_\Sigma
    &=r_0^2\left(\frac{\dot a^2}{2}+a\ddot a\right).
\end{align}
Using \eqref{eq:constant-curvature-first-integral} and
$R_0=6(\ddot a/a+H^2)$ in these relations yields
\begin{align}
    C_1(T(t))
    =\frac{\kappa C r_0^4}{6F_0},
    \qquad
    C_2(T(t))=0.
\end{align}
On every interval where $T(t)$ is locally invertible,
these identities fix $C_1(v)$ and $C_2(v)$ and establish
\eqref{eq:constant-curvature-matched-mass}.}
\end{proof}

{\begin{proposition}[Null-energy obstruction]
\label{prop:constant-curvature-nec}
Under the hypotheses of
Proposition~\ref{prop:constant-curvature-exterior},
the interior and exterior matter tensors both satisfy
the null energy condition if and only if $C=0$.
Consequently, a nonzero radiation component in the
constant-trace interior cannot be matched within this
class while preserving the null energy condition on
both sides.
\end{proposition}
\begin{proof}
For the interior perfect fluid,
\begin{align}
    \rho_-+p_-=\frac{4C}{3a^4}.
\end{align}
Thus the interior null energy condition requires $C\geq0$.\\
For constant $F_0$, the exterior field equations give
\begin{align}
    \kappa\mu
    &=\frac{2F_0M_{,v}}{r_v^2},
    \\
    \kappa\rho
    &=\frac{2F_0M_{,r_v}}{r_v^2}
      -\frac{F_0R_0-f(R_0)}{2},
    \\
    \kappa p
    &=-\frac{F_0M_{,r_vr_v}}{r_v}
      +\frac{F_0R_0-f(R_0)}{2}.
\end{align}
Substituting
\eqref{eq:constant-curvature-matched-mass}, one obtains
\begin{align}
    \mu &=0,
    \\
    \rho
    &=-\frac{T_0}{4}
      -\frac{Cr_0^4}{3r_v^4},
    \\
    p
    &=\frac{T_0}{4}
      -\frac{Cr_0^4}{3r_v^4}.
\end{align}
Here the radial pressure is $p_r=-\rho$, while $p$
is the tangential pressure. The exterior null energy
condition therefore reduces to
\begin{align}
    \rho+p
    =-\frac{2Cr_0^4}{3r_v^4}\geq0,
\end{align}
which requires $C\leq0$.
Combining the interior and exterior conditions gives $C=0$.\\
Conversely, when $C=0$, both matter tensors are
$T_{\mu\nu}=(T_0/4)g_{\mu\nu}$ and saturate the null
energy condition. The matched mass becomes
$M=R_0r_v^3/24$, and both regions are locally portions
of the same maximally symmetric spacetime.
Thus this case contains no nontrivial material collapse.
\end{proof}}

{The above analysis reveals a strong structural constraint on the
scalar degree of freedom in metric $f(R)$ gravity within the
generalized Vaidya class and the exterior matter sector considered
here. In particular, \eqref{eq:derf-form} implies that the scalaron
is at most linear in the areal radius along each null slice.
Its radial configuration is therefore restricted to the form
determined by $A(v)$ and $B(v)$, excluding localized radial
scalaron profiles within this class.}

{
This has direct physical consequences. In the scalar--tensor
description, the scalaron is a propagating degree of freedom.
Here, however, the assumed metric and matter sector constrain
its radial configuration, leading to two distinct branches:
\begin{itemize}
\item If $A(v)\neq0$, the scalaron cannot approach a finite constant
at large radius. For branches with $f_{,R}>0$ and $f_{,R}$ finite
at finite curvature, this excludes a globally extended exterior
with finite asymptotic curvature: the branch either fails to
continue to arbitrarily large radius or requires unbounded
curvature. In the Starobinsky model, the formal continuation
also produces divergent matter variables and, for $\alpha>0$,
violates the null energy condition.
\item If $A(v)=0$ throughout a matched interval, the normal
junction condition forces $B'(v)=0$. Assuming $f_{,RR}\ne0$,
the curvature is constant and the interior matter has constant
trace. Nontrivial homogeneous dust collapse is therefore
excluded. More generally, constant-trace interiors admit
formal matched solutions containing a radiation-like component
and vacuum energy. However,
Proposition~\ref{prop:constant-curvature-nec} shows that imposing
the null energy condition on both sides forces the radiation
component to vanish. The remaining matched geometry is locally
maximally symmetric and contains no nontrivial material collapse.
\end{itemize}
Thus, the obstruction to OS dust collapse within this restricted
class follows from the interplay between the junction conditions
and the radial constraint imposed by the exterior field equations,
together with the requirement of a globally extended exterior
with finite asymptotic curvature. This motivates studying more
general exterior geometries or matter sectors. Allowing
$\psi_{,r_v}\neq0$ or $T_{r_vr_v}\neq0$ relaxes the enforced affine
constraint, although the existence and physical admissibility
of the resulting matched solutions require further analysis.}

\section{Conclusion}

In this work, we have revisited the OS collapse problem in metric $f(R)$ gravity, taking into account the additional junction conditions implied by the fourth-order nature of the field equations. In contrast to GR, where matching requires only the continuity of the induced metric and extrinsic curvature, metric $f(R)$ gravity further imposes the continuity of the Ricci scalar and of its normal derivative. These additional conditions severely constrain admissible collapse scenarios. In particular, we have shown that a homogeneous FLRW interior cannot, in general, be matched to a Ricci-flat exterior whenever $f_{,RR}\neq 0$, thereby obstructing the standard OS construction.

Allowing for a generalized Vaidya exterior enlarges the space of admissible configurations. In the absence of restrictions on the exterior matter sector, the matching conditions determine only the boundary data of the mass function, leaving its extension into the bulk undetermined. This leads to a genuine non-uniqueness of exterior solutions and shows that, at a formal level, the obstruction to homogeneous collapse can be relaxed by considering more general exteriors.

However, this degeneracy is removed once the exterior matter content is restricted to the generalized Vaidya form. In that case, the field equations impose a strong structural constraint, Eq.~\eqref{eq:derf-form}, which drastically reduces the functional freedom of the exterior solution. For the Starobinsky model, this condition is sufficient to determine the exterior mass function uniquely in terms of the interior scale factor. More generally, for locally invertible $f_{,R}$ with $f_{,RR}\neq 0$, the matching data uniquely fix the exterior solution on intervals where the boundary map is locally invertible.

{For branches with $f_{,R}>0$, $f_{,RR}\ne0$, and $f_{,R}$
regular and finite at finite curvature, Corollary~1 excludes
a globally extended exterior with finite asymptotic curvature
when $A(v)\ne0$. The remaining branch, $A(v)=0$, forces
constant curvature and restricts the interior matter to a
radiation-like component plus vacuum energy.
Although this branch admits formal matched solutions,
Proposition~\ref{prop:constant-curvature-nec} shows that the
interior and exterior null energy conditions require opposite
signs of the radiation integration constant $C$. Both conditions
can therefore hold only for $C=0$, leaving a locally maximally
symmetric geometry supported solely by the vacuum-energy
component.\\
Consequently, within the generalized Vaidya metric class
$\psi=0$ and the exterior matter sector $T_{r_vr_v}=0$,
the obstruction extends beyond dust: under the stated
regularity and asymptotic assumptions, nontrivial collapse
of a spatially flat FLRW perfect-fluid interior is excluded
when the matter on both sides satisfies the null energy
condition. Local extensions and matchings involving more
general exterior geometries or matter sectors remain open.}

{
Allowing $\psi_{,r_v}\neq0$ modifies the radial field equation \eqref{eq:Frvrv},
so that $T_{r_vr_v}=0$ no longer forces $f_{,R}$ to be affine in
$r_v$. Whether this larger metric class admits physically acceptable
matched collapse solutions requires further analysis.
}

% \sout{
% Despite this apparent mathematical resolution, the resulting solutions are not physically viable in generic $f(R)$ models. The branch with $A(v)\neq 0$ leads to an unbounded growth of the scalar degree of freedom, resulting in divergent curvature and non-localized matter distributions at large radius. On the other hand, the branch $A(v)=0$ suppresses the radial dynamics of the scalaron and forces the interior onto a constant-curvature sector, thereby excluding nontrivial dust collapse and restricting the admissible matter content to constant-trace sources.

% These results indicate that the obstruction to OS-type collapse in metric $f(R)$ gravity is not merely a consequence of over-constrained junction conditions, but reflects a deeper incompatibility between homogeneous collapse and the allowed configurations of the scalar degree of freedom in null-radiating geometries. In particular, the generalized Vaidya ansatz constrains the scalaron in a way that prevents it from supporting a physically acceptable exterior configuration.

% We therefore conclude that, within the restricted generalized Vaidya matter sector considered here, the OS dust collapse problem in metric $f(R)$ gravity does not admit a physically viable resolution. This suggests that any consistent description of gravitational collapse in these theories may require a more general class of exterior geometries, in which the scalar degree of freedom is treated dynamically rather than being constrained by an effective fluid description.}

\ack{%
The work of R.G. is supported by ANID FONDECYT Regular Grants Nos.~1262002 and 1262349 (Chile). S.C. acknowledges the IUCAA for providing the facility and support under the visiting associateship program. Acknowledgment is also given to the Vellore Institute of Technology for the financial support through its Seed Grant (No. SG20230027), 2023.
}

\bibliographystyle{iopart-num}
\bibliography{bibliography}

\end{document}